\documentclass[twocolumn,amsthm]{autart}  

\usepackage{graphics,epstopdf,wrapfig}
\usepackage{epsfig}
\usepackage{subfig}
\usepackage{amsmath}
\usepackage{mathrsfs}
\usepackage{xcolor}
\usepackage{epsfig}
\usepackage{multirow}
\usepackage{amssymb,latexsym,amsfonts,amsmath}
\usepackage{graphicx}
\usepackage{color}
\usepackage{tikz}
\usetikzlibrary{automata,positioning,shapes,arrows}
\usepackage{epstopdf}
\usepackage{tikz}
\usepackage{mathtools}
\usepackage{verbatim}
\usepackage{tikz-3dplot}
\usetikzlibrary{shapes.geometric}
\usepackage{environ}
\NewEnviron{myequation}{%
\begin{equation}
\scalebox{1}{$\BODY$}
\end{equation}
}

\usepackage{cite}
\usepackage[authoryear]{natbib}
\usepackage{adjustbox}
\tdplotsetmaincoords{70}{30}
\usepackage{mathtools}
\usepackage{algorithm} 

\usepackage[noend]{algorithmic}  
\usetikzlibrary{intersections}
\usetikzlibrary{patterns}
\usetikzlibrary{automata,positioning,shapes,arrows}
\DeclareMathOperator*{\argmax}{argmax} 
\usepackage{tikz}
\usetikzlibrary{automata,positioning,shapes,arrows}
\usetikzlibrary{shapes,snakes}
\tikzstyle{block} = [draw, rectangle, minimum size=3em]
\tikzstyle{bigblock} = [draw, rectangle, minimum height=7em, minimum width=11em]
\theoremstyle{remark}
\tikzset{>=latex}
\newtheorem{lemma}{Lemma}

\newtheorem{example}{Example}

\newtheorem{definition}{Definition}
\newlength\myindent
\makeatletter
\let\AND  \@undefined
\let\endAND  \@undefined
\makeatother

\begin{document}
\begin{frontmatter}
\title{Constrained Active Classification Using Partially Observable Markov Decision Processes}


\author[UT]{Bo Wu}\ead{bowu86@gmail.com},
\author[UCB]{Niklas~Lauffer}\ead{nlauffer@berkeley.edu}, 
\author[UT]{Suda Bharadwaj}\ead{suda.b@utexas.edu},
\author[CalTech]{Mohamadreza Ahmadi}\ead{mrahmadi@caltech.edu},
\author[ASU]{Zhe~Xu}\ead{xzhe1@asu.edu},    
\author[UT]{Ufuk Topcu}\ead{utopcu@utexas.edu}

\address[UT]{the Department of Aerospace Engineering and Engineering Mechanics, and the Oden Institute for Computational Engineering and Sciences, University of Texas, Austin, Austin, TX 78712} 
\address[UCB]{the Department of Electrical Engineering and
Computer Sciences at the University of California, Berkeley, Berkeley,
CA 94720} 
\address[CalTech]{the Center for Autonomous Systems and Technologies
at the California Institute of Technology, 1200 E. Calif. Blvd., MC 104-
44, Pasadena, CA 91125} 
\address[ASU]{School for Engineering of Matter, Transport and Energy, Arizona State University, Tempe, AZ 85287} 

\maketitle

\begin{abstract}
In this work, we study a problem of actively classifying the attributes of dynamical systems characterized as a  finite set of  Markov decision process (MDP) models.  We are interested in finding strategies that actively interact with the dynamical system and observe its reactions so that the attribute of interest is classified efficiently with high confidence. We present a decision-theoretic framework based on partially observable Markov decision processes (POMDPs). The proposed framework relies on assigning a classification belief (a probability distribution) to the attributes of interest. Given an initial belief, confidence level over which a classification decision can be made, a cost bound, safe belief sets, and a finite time horizon, we compute POMDP strategies leading to classification decisions. We present three different algorithms to compute such strategies. The first algorithm computes the optimal strategy exactly by value iteration. To overcome the computational complexity of computing the exact solutions, we propose a second algorithm based on adaptive sampling and third based on Monte Carlo tree search to approximate the optimal probability of reaching a classification decision. We illustrate the proposed methodology using examples from medical diagnosis, security surveillance, and wildlife classification.
\end{abstract}

\end{frontmatter}

\section{Introduction}\label{sec:introduction} 

Active classification \cite{hollinger2017active} is a sequential decision-making process that is inherently ``curious'', i.e., it has control over the data acquisition  and interacts closely with the object of interest. 
It is a plausible approach in many practical classification applications, such as  medical diagnosis \cite{hauskrecht2000planning,ayer2012or} and target tracking \cite{adhikari2018applying}. In these scenarios, the object of interest has certain underlying dynamics and the classification decision has to be made efficiently with high accuracy.

{As illustrated in Figure~\ref{fig:classification}, active classification is  a closed-loop process where the classification process dynamically makes queries, i.e., selects different actions that may directly affect the state evolution of the object. It then obtains observations to update its confidence of correct classification decision. A classification decision can be made when the error probability meets a given threshold. Otherwise, additional queries (actions) will be scheduled to collect more evidence (observations) to assist the classification process.}

\begin{figure}[t]
	\centering

		\begin{tikzpicture}[shorten >=1pt,node distance=3cm,on grid,auto, bend angle=20, thick,scale=0.65, every node/.style={transform shape}] 
		\node[bigblock,dotted] (s0)   {}; 
		\node [draw, diamond, aspect=1.5]  (s1) [right =   4 cm of s0] {~~~~~~~~~~~~~~~~~~~~~~~~~}; 	
		\node[draw,align=left] (s2) [ right = 4 cm  of s1]  {Classification\\ decision}; 
			\node[draw,align=left] (s3) [ below left = 2 cm and 1.5 cm  of s1]  {Active\\ classification}; 
		

        \node at (0,0){\includegraphics[scale=0.2]{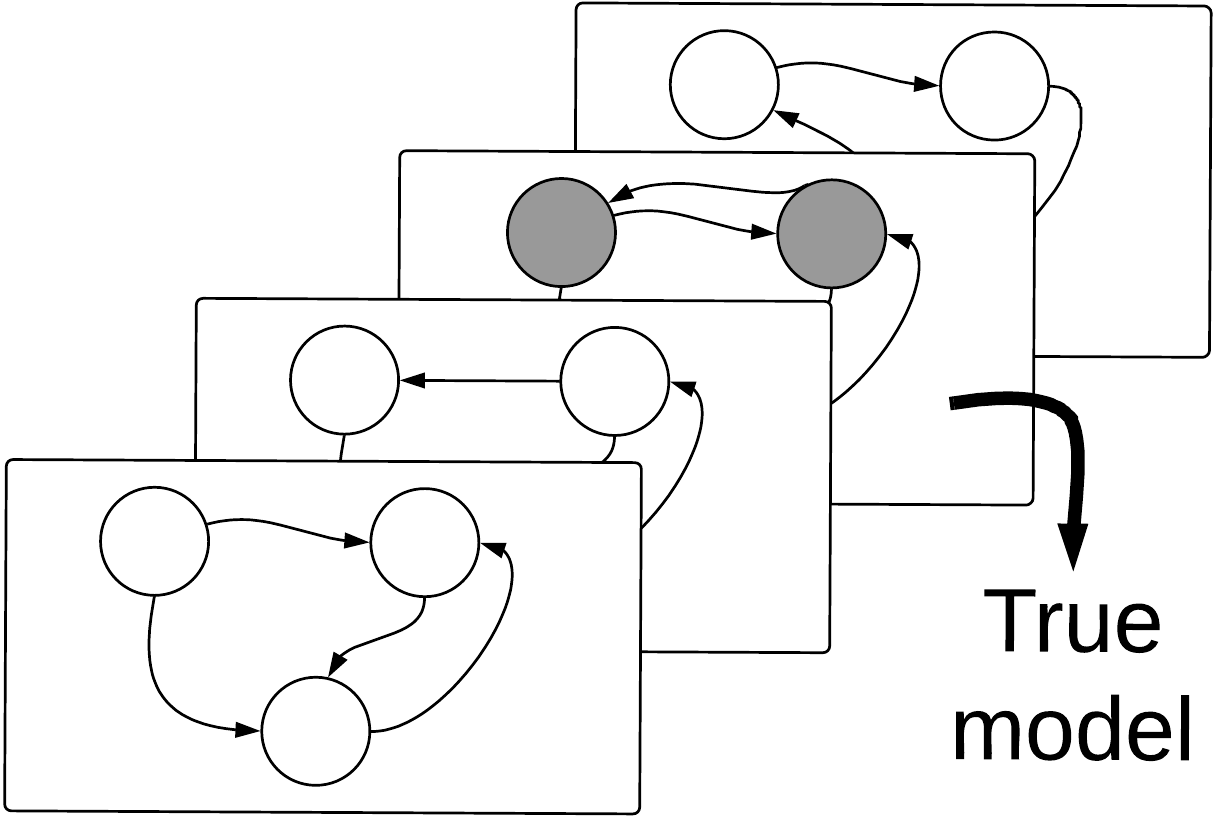}};
			
		\draw [->](0,1.5) -- (0,1.8)--(4,1.8)--(4,1.1);
		\draw [->](4,-1.2)--(4,-2)--(3.5,-2);
		\draw [->](1.5,-2)--(0,-2)--(0,-1.5);
    	\node[text width=3cm] at (-2.2,1){A family $\mathcal{C}$ of \\ candidate models };
    	\node[text width=3cm] at (1.8,-1.8){Action};
    	\node[text width=3cm] at (3,2){Observation};
    	\node[text width=3cm] at (7.5,0.2){Yes};
    	\node[text width=3cm] at (5.7,-1.5){No};
    	\node[text width=3cm] at (0.9,-1.2){({\textbf{1}},1)};
    	\node[text width=3cm] at (1.6,-0.7){({\textbf{1}},2)};
    	\node[text width=3cm] at (2.3,-0.2){({\textbf{2}},1)};
    	\node[text width=3cm] at (2.9,0.3){({\textbf{2}},2)};
    	\node[text width=3cm] at (4.25,0){Misclassification \\ probability met?};
    	\node[text width=3cm] at (-2.5,-1.5){$\mathcal{C}={\mathbf{\mathcal{C}_1}}\times\mathcal{C}_2$\\$\mathcal{C}_1=\{1,2\}$\\$\mathcal{C}_2=\{1,2\}$};
		\path[->]
				(s1) edge node [above] {} (s2)  

		; 
		\end{tikzpicture} 
	
	\caption{Active classification from a family of candidate models. $\mathcal{C}_1$ in bold is the attribute to be classified. }
	\label{fig:classification}
\end{figure}

{In this paper, we study the problem of active classification in which the object whose attributes to be classified is a dynamical system.} In order to  capture the stochastic uncertainties of the outcomes associated with each action during classification, we assume that the underlying dynamic model of the object belongs to a family of known Markov decision processes (MDPs) \cite{puterman2014markov}. The states of the underlying MDP is fully observable but what is not known is which MDP the object corresponds to. {Such an unknown variable may consist of a number of attributes, where some of the attributes are of interest to be classified. For example, as shown in Figure \ref{fig:classification}, there are two attributes $\mathcal{C}_1$ and $\mathcal{C}_2$, each could take two possible values $1$ and $2$. As a result, there are four candidate models. The classification objective is to determine whether $\mathcal{C}_1$ is 1 or 2.} Furthermore, each action may incur a certain cost, for example the cost of test or treatment in medical diagnosis. Therefore,  the overall cost of classification must  be bounded.

Naturally, such an active classification problem  can be cast into the framework of hidden model MDP (HMMDP), which is a special case of partially observable Markov decision process (POMDP) \cite{kaelbling1998planning}. {Since the true underlying model is not known, it is only possible to maintain a \emph{belief} of which model is the true one for making the classification decision. The belief, in the proposed setting, is a probability distribution over the possible MDP models and evolves  based on the history of observations and actions}.

We are interested in providing a guaranteed bound on the classification error. The classification decision is made whenever the probability of  the attribute of interest being the true one exceeds a  given threshold based on the misclassfication probability, as shown in Figure \ref{fig:classification}. Meanwhile, we impose additional requirements to maintain the safety or privacy of the object to be classified, as considered in \cite{ahmadi2018privacy,ahmadi2019control}. A classification decision should be reached without violating them. 

Given the POMDP model, the desired thresholds on the misclassification probability and the accumulated costs, we propose two approaches to obtain the classification strategy that optimizes the probability to reach a classification decision. The exact solution is inspired by policy computation for constrained MDPs \cite{el2018controlled}, especially cost-bounded reachability in MDPs \cite{andova2003discrete}. It relies on obtaining the underlying belief MDP whose states are beliefs from the original POMDP. 
Then the optimal strategy can be computed from the obtained belief MDP. To overcome the computational complexity of finding exact solutions for POMDPs, our second approach  adaptively samples the actions in the belief MDP to approximate the optimal probability to reach a classification decision.

The extensions compared with the preliminary results in \cite{wu2018cost} are as follows. First, instead of only deciding which true model the underlying MDP belongs to in \cite{wu2018cost}, this paper considers classifying an attribute of interest out of several attributes that a model may have. Therefore, the classification objective in \cite{wu2018cost} is a special case of this paper. Second,  this paper considers additional safety constraints during the classification process, whereas \cite{wu2018cost} focused merely on a pure reachability problem in the belief space. Therefore, new and more detailed proofs are presented to show the (approximate) optimality of the probability to reach a classification decision. Third, a new Monte Carlo tree search based approach is proposed for solving the problem on the go while lowering the computation complexity. 

The contribution of this paper is three-fold. To the best of our knowledge, we are the first to consider active classification in an HMMDP modeling framework. Second, besides the classification objective, we also present additional constraints regarding cost bound, safety, and privacy. Third, we establish new algorithms, especially an adaptive sampling and Monte Carlo tree search method, to solve the proposed classification problem with constraints. We also prove the asymptotic convergence of the algorithm's output to the optimal probability to reach a classification decision.
\subsection{Related Work}
POMDPs  have been used in a variety of applications in medical diagnosis \cite{hauskrecht2000planning,ayer2012or}, health care \cite{zois2013energy}, privacy \cite{wu2018privacy, savas2022entropy}, and robotic planning \cite{wu2021supervisor,hollinger2017active,myers2010POMDP,spaan2008cooperative,chen2012constrained}. Many existing results focus on minimizing the costs as well as the classification uncertainty measured in terms of entropy~\cite{fox1998active}. The classification accuracy is often implicitly embedded in the rewards. However, entropy is an uncertainty measure that is not directly translated to a guaranteed classification accuracy  \cite{settles.tr09}. In \cite{spaan2015decision}, a robot sensing and surveillance scenario is considered in a POMDP framework, where the objective is to reach a high certainty level in the sensing task while balancing the surveillance task. However, they did not consider additional cost bound and safety constraints as in this paper.

A similar problem for robot planning was studied recently in \cite{Wang:2018:BPS:3237383.3237424}, where the belief in a POMDP must reach some goal states. Their planning problem is in a constrained belief space in which the reachability to goal states is guaranteed with probability one. However, it could be over-conservative in many cases since reaching goal states could never be guaranteed with probability one and, therefore, that constrained belief space may not even exist.

The computation complexity of finding an optimal policy in a POMDP is prohibitively high, which motivates many approximate solutions. One of the most popular approaches is the Monte Carlo tree search (MCTS) method \cite{coulom2006mcts,browne2012survey}, for example, the Partially Observable Monte Carlo Planning (POMCP) algorithm proposed in \cite{silver2010monte}. We introduce a version of Monte Carlo tree search along with an adaptive sampling approach, an offline MCTS method \cite{fu2018monte}, that can both compute an approximate solution with significantly reduced computation time.



\section{Preliminaries}\label{sec:preliminaries}
In this section, we describe preliminary notions and definitions used in the sequel, following \cite{puterman2014markov,chades2012momdps}.

\subsection{Markov Decision Processes}
 Formally, a Markov decision process (MDP)  is defined as follows.
\begin{definition}
	An MDP $\mathcal{M}$ is a tuple $\mathcal{M}=(S,\hat{s},A,T,C)$ where $S$ is a finite set of states, $\hat{s}$ is the initial state, $A$ is a finite set of actions, $T:S\times A\times S\rightarrow [0,1]$ is a probabilistic transition function with $
		T(s,a,s'):=P(s'|s,a),\text{ for all }    s,s'\in S, a\in A$;  $C:S\times A\rightarrow\mathbb{R}_{\geq 0}$ is a cost function.
\end{definition}
\subsection{Hidden-Model MDPs}\label{subsec:hmmdps}
The classification problem stated in Section \ref{sec:introduction} assumes that the object of interest is an unknown MDP that belongs to a known finite set $\mathcal{C}$ of MDPs. Formally, such a problem can be modeled in the framework of hidden model MDP (HMMDP) where the   underlying true MDP is from a finite set  $\mathcal{M}=\{\mathcal{M}_i|i\in\mathcal{C}\}$, $\mathcal{C}=\mathcal{C}_1\times\mathcal{C}_2...\times\mathcal{C}_m$. Therefore, each model has $m$ attributes and each attribute is a finite set. For example, $\mathcal{C}_1$ could refer to gender and  $\mathcal{C}_2$ could refer to age.

We assume, without loss of generality, that the MDPs share the same state space $S$, initial state $\hat{s}$, the action set $A$ and cost function $C$.
\begin{example}\label{exp:1}
Consider a medical diagnosis scenario, where each MDP $\mathcal{M}_i$ captures how stages for a disease $i$ evolve based on tests and treatments. In particular, there are three states in each MDP, and the states in the MDP represent the early stage, medium stage, and late stage of the disease, with increasing order of severity. There is a family of $|\mathcal{C}|$ diseases out of which one of them is true that needs to be diagnosed and treated. Therefore, there is only one attribute in $\mathcal{C}$. The action space includes $|\mathcal{C}|+1$ actions, namely $a_i,i\in\{1,..,|\mathcal{C}|\}$ for treatment $i$ and $a_{|\mathcal{C}|+1}$ for doing nothing but observe. Each treatment $a_i,i\in\{1,..,|\mathcal{C}|\}$ is more effective on the disease $i$ with a cost $C(s,a_i)$ and  $a_{|\mathcal{C}|+1}$ will introduce no cost. Each treatment or passive observation will introduce a probabilistic transition between different stages of the disease. Figure \ref{fig:hmmdp} shows all the possible transitions while omitting the transition probabilities for simplicity.
\end{example}
		\begin{figure}[t]
			\centering
				\begin{tikzpicture}[shorten >=1pt,node distance=3cm,on grid,auto, bend angle=20, thick,scale=1, every node/.style={transform shape}] 
				\node[state,initial] (s0)   {$s_1$}; 
				\node[state] (s1) [right= of s0] {$s_2$}; 	
				\node[state] (s2) [right= of s1]  {$s_3$}; 
		        \node[text width=3cm] at (0.8,-0.8) {stage 1};
				\node[text width=3cm] at (3.8,-0.8) {stage 2};
				\node[text width=3cm] at (6.8,-0.8) {stage 3};
 	
				\path[->]
				(s0) edge [loop above] node {}()
				(s0) edge [bend left] node {} (s1)
				(s1) edge [loop above] node {}()
				(s1) edge [bend left] node {} (s0)  
				(s1) edge node {} (s2)  
				(s2) edge [loop above] node {}()
				; 
				\end{tikzpicture} 
			
			\caption{Feasible state transitions in MDPs for Example \ref{exp:1}.}
\vspace{-4mm}
			\label{fig:hmmdp}
		\end{figure}
Given the initial state $\hat{s}$, we denote $\hat{b}_{\hat{s}}(i)$ as the initial  probability that $\mathcal{M}_i$ is the  underlying true model. For Example \ref{exp:1},  $\hat{b}_{\hat{s}}(i)$ denotes the initial likelihood of the true disease being disease $i$.  Then an HMMDP is essentially a partially observable Markov decision process (POMDP) $\mathcal{P}=(S\times\mathcal{C},\pi,A,T,Z,O,C)$ where $S\times\mathcal{C}$ is s finite set of states;
	$\pi:S\times\mathcal{C}\rightarrow[0,1]$ is the initial state distribution with $\pi(s,i)=\hat{b}_{\hat{s}}(i),  \text{if } s=\hat{s}$ and $
		0  \text{ otherwise;}$
		$A$ is a finite set of actions that is the same as the underlying MDPs;
		\ $T$ is given as follows,
		\begin{equation}\nonumber
		\begin{split}
		    T((s,i),a,(s',i'))&= T_i(s,a,s')  \text{ if } i=i',\\
		                      &= 0  \text{               otherwise};
		\end{split}
		\end{equation}
		$Z$ is the set of all possible observations;
		$O:S\times\mathcal{C}\times Z\rightarrow [0,1]$ is the observation function that deterministically identifies the state (but not the model element); and
		$C:S\times A\rightarrow\mathbb{R}_{\geq 0}$ is the cost function.

The definition of $T$ implies that the underlying true model $\mathcal{M}_i$ will not change to any other model during the classification process. 

It is essential to keep track of both the accumulated cost $e\in E\subset\mathbb{R}_{\geq 0}$ and a \emph{belief}  $b$ where $b(s,i)\in[0,1],\sum_s\sum_{i\in\mathcal{C}}b(s,i)=1$, which is a probability distribution over all the possible HHMDP states. From the definition of the observation function $O$, it can be seen that the observation gives  perfect information about the state element $s$ in the state-model tuple $(s,i)$, but not the model element $i$. Therefore, when $z$ is observed, we denote $b(s,i)=b_s(i)$ for simplicity. The belief space is reduced to $B\subseteq \mathcal{R}^{|\mathcal{C}|}$. We can then obtain a belief MDP $\mathcal{B}=(Q=B\times E,\hat{q}=(\hat{b}_{\hat{s}},0),A,T)$, where
	\begin{itemize}
        \item a state is denoted as $q=(b,e)$, where $b\in B$ is the belief and $e\in E$ is the accumulated cost so far;
		\item The state transition probability is described by \begin{equation}\label{equation:HMMDP transition function}  
		\begin{split}
T(q,a,q')&=T((b,e),a,(b',e'))= T(b_{s},a,b'_{s'})\\&=\sum_i b_s(i)T_i(s,a,s'),	    
		\end{split} 
		\end{equation} 
  \vspace{-0.4in}
		\begin{align}\label{equation:HMMDP belief update}
b'_{s'}(i)&=\frac{O((s,i),z)T_i(s,a,s')b_s(i)}{\sum_j O((s,j),z)T_j(s,a,s')b_s(j)} \\ e'&=e+C(s,a).	
		\end{align}

	\end{itemize}

A state-action path $\omega$ of length $N$ in the belief MDP $\mathcal{B}$ is of the form $\omega=(b_{s_0},e_0)a_0(b_{s_1},e_1)a_1...(b_{s_N},e_N)$. The accumulated cost $C(\omega)$ along the path $\omega$ is given by 
$
  C(\omega)=e_N.
$
At each state of the belief MDP, the choice of actions is determined by a policy $\mu=\{\mu_i|\mu_i:B\times E\rightarrow A,0\leq i\leq H \}$. Given a strategy $\mu$ to resolve the nondeterminism in the action selection of the HMMDP model, it is possible to calculate the probability  of the occurrence of a path $\omega$ by 
$
P(\omega)=\prod_i T(b_i,a_i,b_{i+1}),
$
where $\mu_i(b_{s_i},c_i)=a_i$.
\section{Problem Formulation}\label{sec:Problem Formulation}
{We are interested in classifying a particular attribute $\mathcal{C}_j$ where $j=\{1,...,m\}$ of the underlying MDP. Suppose $\mathcal{C}_j=\{1,2,\dots,k\}$, i.e., it takes one of $k$ values.  To make a classification decision within a finite time bound $H$,  we keep track of the belief $b_s$ and claim the attribute $\mathcal{C}_j$ of underlying model is $i$, whenever }
\begin{equation}\label{equation:reachability}
{\sum_{\mathcal{C}_j(k)=i}b_s(k)\geq\lambda_i,}
\end{equation}
{where $\mathcal{C}_j(k)=i$ represents the attribute $\mathcal{C}_j$ of the MDP model $\mathcal{M}_k$ being $i$,   $\lambda_i\in(0.5,1]$ denotes the minimum confidence to claim that the attribute $\mathcal{C}_j$ belongs to $i$.}

{On the other hand, reaching a classification decision may not be the only objective. For example, it may also be desired to reach a diagnosis decision at the early or intermediate stages of a disease. In some applications, it is also important to preserve some secret as characterized in belief space from being learned during active classification. Therefore, the belief may have to be constrained within a set of \emph{safe} belief $B_{safe}\subseteq B$ before the classification decision is reached. In such cases, the set of goal states $G$ is defined to be $G=\cup_{i\in\mathcal{C}} G_i$ where $G_i=\{b|b\in B_{safe}, \sum_{\mathcal{C}_j(k)=i}b_s(k)\geq\lambda_i\}$.} 
A belief state in $G$ is a terminal state in the belief MDP. Once the terminal state is reached, the classification task is accomplished. 

In many practical applications, it is also essential to accomplish the classification task with a fixed amount of cost. That is, for   state-action path $\omega$ in the belief MDP $\mathcal{B}$ where $\omega=(b_{s_0},e_0)a_0(b_{s_1},e_1)a_1...(b_{s_N},e_N)$ such that $b_{s_i}\notin G,i<H$ and $b_{s_N}\in G,N\leq H$, it is required that 
\begin{equation}\label{equation:cost}
   C(\omega)\leq D.
\end{equation} 
Here $\omega$ denotes a path that a classification decision is met for the first time while the beliefs stay inside of $B_{safe}$. We denote $\Omega_G$ as the set of such paths that reach $G$ within time bounds $H$ and cost bound $D$ while remaining in the safe sets. 

For the classification task, the objective is to compute a policy $\mu$ to dynamically select classification actions.  
With $\mu$, we can get the probability $P^{\leq D}_{\mu}(B_{safe}U^{\leq H} G)$ to reach a classification decision within time bound $H$ and cost bound $D$, where
$
P_\mu^{\leq D}(B_{safe}U^{\leq H}G)=\sum_{\omega\in\Omega_G}P(\omega).
$
{Put the pieces together, the active classification problem aims to compute a policy $\mu^*$ such that}
\begin{equation}{\label{equation:safe reachability}}
    {\mu^*=\argmax_\mu P^{\leq D}_{\mu}(B_{safe}U^{\leq H} G).}
\end{equation}
\section{Cost-Bounded Active Classification}\label{sec:Cost-Bounded Active Classification}
In this section we introduce two approaches to solve the active classification problem as defined by (\ref{equation:safe reachability}).

\subsection{Cost-bounded unfolding}
In this approach, the first step is to obtain a finite belief-state MDP $\mathcal{B}$ from the HMMDP model considering the accumulated cost. Such a procedure is called \emph{unfolding} and inspired by the similar treatment in MDPs \cite{andova2003discrete} for cost-bounded properties. In this paper, we extend this procedure to POMDPs.

Algorithm \ref{alg:unfold} describes how to obtain $\mathcal{B}$. It is a recursive breadth-first traversal starting from the initial  state $\hat{q}=(\hat{b}_{\hat{s}},0)$ with initial belief $\hat{b}_{\hat{s}}$ and zero accumulated cost. The algorithm goes on for $H$ iterations or until there is no more state to explore,   as can be seen in Line \ref{algorithm:terminate condition}. At each iteration $i$, we iterate through every state $q=(b,e)$ (Line \ref{algorithm:q}), where $b$ is the belief state, $e$ is the cost accumulated so far. For each action $a\in A$ (Line \ref{algorithm:action} ), we calculate its next accumulated reward $e'$ (Line \ref{algorithm:cost}). We terminate further exploring the successors of this state if $e'>D$ (Line \ref{algorithm:cost violation}), i.e., when the cost bound is exceeded. Otherwise,  the successor belief state $b'$ is computed (Line \ref{algorithm:belief}) as well as the transition probability (Line \ref{algorithm:transition}). {If $b'\in G$,  a  classification decision is reached, otherwise if $b'\in B_{safe}$ as shown in Line \ref{algorithm:safe or not},  $q'=(b',e')$ is added to a set $Next$ and will be expanded in the next iteration (Line \ref{algorithm:next state}). By construction, a state is a terminal state if it violates $B_{safe}$ or the cost bound or its belief component belongs to $G$.}

\begin{figure}
		\begin{algorithm}[H]
  \begin{algorithmic}[1]
			\STATE{\textbf{INPUT:}{  An HMMDP model with MDPs $\mathcal{M}_i,i\in\mathcal{C}$, $\mathcal{C}=\{1,...,L\}$, time bound $H$, cost bound $D$, initial belief $\hat{b}_{\hat{s}}$ and reachability constraint as defined in (\ref{equation:reachability}).}}
		\STATE{\textbf{OUTPUT:} 	Finite state MDP $\mathcal{B}=(Q=B\times E,\hat{q},A,T)$, $Q=\{(\hat{b}_{\hat{s}},0)\},Cur=\{(\hat{b}_{\hat{s}},0)\},i=0$}
			\WHILE {{$i<H$ and $Cur\neq\emptyset$ \label{algorithm:terminate condition}}}
			        \STATE {$Next=\{\}$}
    			    \FOR{$q=(b_s,e)\in Cur$ \label{algorithm:q}}
    				\FOR{$a\in A$\label{algorithm:action}}
    				\STATE{$e'=e+C(s,a)$ \label{algorithm:cost}}
    				\IF {$e'\leq D$ \label{algorithm:cost violation}}
    				 \FOR{$s'\in S$}
        				\STATE{Compute $b'$ according to (\ref{equation:HMMDP belief update}) \label{algorithm:belief}};
        				\STATE{Let $q'=(b',e')$, $T(q,a,q')$ is computed according to (\ref{equation:HMMDP transition function})\label{algorithm:transition}}
        				\IF {$q'\notin Q$}
        				\STATE {$Q=Q\cup  q'$}
        				 \IF {$b'\notin G$ and $b'\in B_{safe}$}\label{algorithm:safe or not}
        				  \STATE{$Next=Next\cup q'$\label{algorithm:next state}}
        				\ENDIF
                            \ENDIF
        				\ENDFOR
    				\ENDIF
                        \ENDFOR
    			\ENDFOR
					\STATE{ $Cur=Next$}
			\ENDWHILE

			\RETURN{$\mathcal{B}$}
			\end{algorithmic}
			\caption{Cost-Bounded Unfolding}\label{alg:unfold}	
   
		\end{algorithm}
	\end{figure}

From the output of Algorithm \ref{alg:unfold}, it can be observed that the accumulated cost is already encoded in the state space of $\mathcal{B}$. Once $\mathcal{B}$ is obtained, it is then possible to calculate the optimal strategy $\mu^*$ on $\mathcal{B}$ to achieve the following probability
\begin{equation}\label{equation:reachability1}
{P^{\leq D}_{max}(B_{safe}U^{\leq H} G)=P_{\mu^*}(B_{safe}U^{\leq H}G),}
\end{equation}
i.e., the maximized probability to reach a classification decision within $H$ steps but without considering the cost bound. 

To get $\mu^*$, it is needed to compute the maximal probability, denoted as $P^q_{max}(B_{safe}U^{\leq k }G)$ to reach $G$ with in $k\in\{0,...,H\}$ steps from any $q\in Q$ where
$
P^{\hat{q}}_{max}(B_{safe}U^{\leq H }G) = P_{max}(B_{safe}U^{\leq H} G).
$
{We first divide the state set $Q$ into three disjoint subsets  $Q^{yes}=\{q=(b,e)|b\in G\}$, $Q^{no}=\{q=(b,e)|b\notin B_{safe}\}$ and $ Q^{?}=Q\backslash (Q^{yes}\bigcup Q^{no})$. The computation of $P^q_{max}(B_{safe}U^{\leq k }G)$ is essentially a  dynamic program as shown below.} 
\begin{align}
    &{P^q_{max}(B_{safe}U^{\leq i} G)=1, ~~~\forall q\in Q^{yes},i\in\{0,...,H\},}\nonumber\\
    &{P^q_{max}(B_{safe}U^{\leq i} G)=0, ~~~\forall q\in Q^{no}, i\in\{0,...,H\},}\nonumber\\
    &{P^q_{max}(B_{safe}U^{\leq 0} G)=0, ~~~\forall q\in Q^{?},}\nonumber\\
    &{P^q_{max}(B_{safe}U^{\leq i} G)=}\nonumber\\
    &{\max_{a\in A} \sum_{q'\in Q}T(q,a,q')P^{q'}_{max}(B_{safe}U^{\leq i-1}G), \forall q\in Q^?,i\neq 0.}\nonumber
\end{align}
{Then it can be seen that  for $i\in\{1,...,H\}$,}
$${\mu^*_i(q)=\argmax_{a\in A}\sum_{q'\in Q}T(q,a,q')P^{q'}_{max}(B_{safe}U^{\leq i-1}G).}$$

\subsection{ Adaptive Sampling in Belief Space}

The exact solution requires constructing the belief MDP $\mathcal{B}$. It leads to the curse of dimensionality in the computation of $P_{max}(B_{safe}U^{\leq H}G)$. 
In this subsection, we propose an alternative approach  using an offline MCTS method inspired by the adaptive multi-stage sampling algorithm (AMS) algorithm proposed in \cite{chang2005adaptive} to estimate the optimal classification probabilities. 

The key observations for the active classification problem in (\ref{equation:safe reachability}) that make the AMS a reasonable choice are as follows. First, since the  MDP models in $\mathcal{M}$ are typically  smaller than the belief MDP $\mathcal{B}$, it is easier to simulate sample paths in $\mathcal{B}$ than explicitly specifying $\mathcal{B}$ itself. Furthermore, AMS is  particularly suitable for models where it is unlikely to revisit the same belief state multiple times in a sampled run  \cite{chang2005adaptive}. It is exactly the case for the belief MDP $\mathcal{B}$ obtained with Algorithm \ref{alg:unfold}. It can be observed that in $\mathcal{B}$, the action space remains the same as  the original MDPs. Furthermore, the belief states in $\mathcal{B}$ take  values in a continuous space so that generally it is very rare to revisit the same belief state with the same accumulated cost in a simulated run. Algorithm \ref{alg:ams} shows the belief state sampling procedure, termed as CB-AMS short for cost-bounded adaptive multi-stage sampling.

Algorithm \ref{alg:ams} takes the input of a state $q=(b,e)$ in $\mathcal{B}$, the number of $N_i$ sampling needed, and the current time horizon $i$. It outputs $ \tilde{P}_i^{N_i}(q)$, which is the estimated maximal probability to reach a classification decision  while remaining in the safe belief subset from state $q$ with $H-i$ steps. The initial call of Algorithm \ref{alg:ams} is CB-AMS$((b_0,0),N_0,0)$  for the initial belief $b_0$, the initial cost of $0$, the number $N_0$ of samples needed, and time horizon $0$. 

\begin{figure}[t]
		\begin{algorithm}[H]
 \begin{algorithmic}[1]
			\STATE{\textbf{INPUT:}{  A state $q=(b_s,e)$ in $\mathcal{B}$, the number of samples $N_i$, time horizon $i$.}}
			\STATE{\textbf{OUTPUT:}{The estimated maximal probability $ \tilde{P}_i^{N_i}(q)$.}}
		    \IF{$e>D$ or $i>H$ \label{algorithm:check1} or $b_s\notin B_{safe}$}
			    \RETURN {0}
			\ENDIF
			\IF {$b_s\in G$\label{algorithm:check2}}
			    \RETURN {1}
			\ENDIF
			\FOR{$a\in A$ \label{algorithm:init1}}
			    \STATE {Sample a next state $q'=(b',e')$ by taking action $a$, where $e'=e+C(s,a)$\label{algorithm:sample1}}
			    \STATE {$\tilde{Q}(q,a)=$ CB-AMS($q',N_{i+1},i+1$) \label{algorithm:init2}}
			    \STATE {$N_{a,i}^q=1$}
			\ENDFOR
			\STATE {$n=|A|$\label{algorithm:init3}}
			\WHILE {$n<N_i$ }\label{algorithm:start}
			        \STATE {$a^*=\argmax_a(\frac{\tilde{Q}(q,a)}{N_{a,i}^q}+\sqrt{\frac{2\ln{n}}{N_{a,i}^q}})$\label{algorithm:action selection}}
			        \STATE {Sample a next state $q'=(b',e')$ by taking action $a^*$, where $e'=e+C(s,a^*)$\label{algorithm:sample2}}
			        \STATE {$\tilde{Q}(q,a^*)=\tilde{Q}(q,a^*)+$ CB-AMS($q',N_{i+1},i+1$) \label{algorithm:recur}}
			        \STATE {$N_{a^*,i}^q=N_{a^*,i}^q+1$, $n=n+1$\label{algorithm:end}}
		    \ENDWHILE
            \STATE {$ \tilde{P}_i^{N_i}(q)=\frac{1}{N_i}\sum_a\tilde{Q}(q,a)$}
			\RETURN {$ \tilde{P}_i^{N_i}(q)$}
			\end{algorithmic}
		\caption{Cost-bounded adaptive multi-stage sampling (CB-AMS)}\label{alg:ams}			
		\end{algorithm}
	\end{figure}
	
The algorithm first checks if the termination condition is met at Line \ref{algorithm:check1}. If the accumulated cost or the time horizon exceeds the allowed bound or the belief state $b_s$ violates the safety condition, the algorithm will return  $0$, since the probability to reach a classification decision without violating safety constraints from this state $q$ is $0$. Otherwise, at Line \ref{algorithm:check2}, if the belief $b$ reaches its goal set $G$, the algorithm will return $1$, meaning that the probability to reach a classification decision subject to cost and time bounds, as well as the safety constraints from this state $q$, is $1$.

If the state $q$ is not a terminal state, then the algorithm proceeds to an initialization procedure from Line \ref{algorithm:init1} to Line \ref{algorithm:init2}. It first tries every action $a\in A$ and samples a subsequent state $q'=(b',e')$, where $b'$ is sampled based on the transition probability as defined in (\ref{equation:HMMDP transition function}). At Line \ref{algorithm:init2}, CB-AMS is called  where $\tilde{Q}(q,a)$ denotes the accumulated returned values, and each returned value represents the estimated probability to successfully reach a classification decision from $q$ by executing action $a$). Then we set $N_{a,i}^q$ to be one, where $N_{a,i}^q$ denotes the number of times an action $a$ is sampled from state $q$ at time horizon $i$. This initialization procedure is required  since $N_{a,i}^q$ will be used in Line \ref{algorithm:action selection} where $N_{a,i}^q$ must be nonzero. 

At Line \ref{algorithm:init3}, $n$ denotes  the number of samples collected and is initialized to be $|A|$ since we just tried each action exactly once.  We will then enter the adaptive sampling loop that terminates when the number of samples $n$ reaches $N_i$. In each sampling iteration, we first select an action by the equation defined in Line \ref{algorithm:action selection}. The selection criterion balances between high average return value by the term $\frac{\tilde{Q}(q,a)}{N_{a,i}^q}$ (exploitation) and trying actions that are less sampled by the term $\sqrt{\frac{2\ln{n}}{N_{a,i}^q}}$ (exploration). Once the action $a$ is selected, the algorithm will proceed to sample the next state $q'$ (Line \ref{algorithm:sample2}) and accumulate $\tilde{Q}(q,a)$ by the value returned from calling CB-AMS with the next state as the input (Line \ref{algorithm:recur}). Then $N_{a,i}^q$ and $n$ will both be incremented by one (Line \ref{algorithm:end}). The last step is to average over the accumulated return values and return the result.   

Now we want to analyze the asymptotic performance of Algorithm \ref{alg:ams}. In Algorithm \ref{alg:ams}, we are effectively sampling from the belief MDP $\mathcal{B} $ with the state $q=(b_s,e)$ at time step $i\leq H$.  We denote $R(q,a)$ as the reward by executing action $a$ at state $q$. Note that this reward is not to be confused with the cost function $C$ that represents the classification cost, for example, the test and treatment costs for medical diagnosis. 

Given a strategy $\mu=\{\mu_i|\mu_i:Q\rightarrow A,0\leq i\leq H\}$, the value function $V_i^\mu(q)$ for state $q$ and time step $i$ is
$$
{V^\mu_i(q)=R(q,\mu_i(q))+\sum_{q'\in Q}P(q,\mu_i(q),q')V^\mu_{i+1}(q')}.
$$
We assign $R(q,a)=1$ for all $ a\in A, q=(b_s,e),e\leq D$ and $ b_s\in G$ . Otherwise, $R(q,a)=0$. Once reaching a state $q=(b_s,e)$ with $e>D$ or $b_s\notin B_{safe}$ or $b_s\in G$, the algorithm will return  and such $q$ will not have successive states. From Line \ref{algorithm:check1} and \ref{algorithm:check2},  given a state $q=(b_s,e)$, we know that
\begin{itemize}
    \item $V^\mu_i(q)=0$ if $e>D$ or $b_s\notin B_{safe}$.
    \item $V^\mu_i(q)=1$ if $e\leq D$ and $b_s\in G$.
    \item $V^\mu_{H+1}(q)=0$. 
\end{itemize}   
{It then can found that }
${V^\mu_i(q)=P^{q}(B_{safe}U^{\leq H-i}G),}$
{where $P^{q}(B_{safe}U^{\leq H-i}G)$ denotes the probability to satisfy $B_{safe}U^{\leq H-i}G$ from state $q$. We denote $V_i^*(q)=max_\mu V_i^\mu(q)$.}

{At any time horizon $i$ and state $q$, we denote $U(q')$ as the value returned from calling CB-AMS algorithm for $i+1$ and $q'$. It can be observed that $U(q')$ is a non-negative random variable with unknown distribution and a bounded support. At time step $i$, the sampling process from Line \ref{algorithm:start} to Line \ref{algorithm:end}  can be seen as a one-stage sampling without going further into future time steps where the value function $U(q')$ in state $i+1$ is returned from a black box. Denote}
$$
{U_{max}=\max_{q,a}(R(q,a))+\sum_{q'}P(q,a,q')E[U(q')],}
$$
{which satisifies $U_{max}\leq 1$.}

{The following lemma helps  prove the convergence of our proposed CB-AMS algorithm in Theorem \ref{thm:CB-AMS}}.

\begin{lemma}\label{lemma:1}\cite{chang2005adaptive}
{Given a stochastic value function $U$ over $Q$ with $U_{max}\leq 1$, at any time horizon $i$, state $q$ and $N_i$, define }
$V^*(q)=\max_{a\in A}(R(q,a)+\sum_{q'}P(q,a,q')E[U(q')]),$
then for any $q\in Q$,
$
\lim_{N_i\rightarrow \infty}E[\tilde{P}_i^{N_i}(q)]= V^*(q).
$
\end{lemma}

{Then the following theorem shows that the output of Algorithm \ref{alg:ams} converges to  $P_{max}(B_{safe}U^{\leq H} G)$  as the number of samples $N_i,0\leq i\leq H$, goes to infinity.} 
\begin{thm}\label{thm:CB-AMS}
The algorithm CB-AMS with input $N_i$ for $i=0,...,H$ and an arbitrary initial condition $q\in Q$ satisfies
$$
{\lim_{N_0\rightarrow\infty}\lim_{N_1\rightarrow\infty}\dots\lim_{N_H\rightarrow\infty}E[\tilde{P}_0^{N_0}(q)] = P^q_{max}(B_{safe}U^{\leq H}G)},
$$
{where $P^{q}_{max}(B_{safe}U^{\leq H}G)$ represents  the maximum probability to reach the decision region $G$ in $H$ steps with costs no larger than $D$ from state $q$ while staying inside of $B_{safe}$}.
\end{thm}
\begin{proof}
We prove the theorem by a backward inductive argument. Given a time bound $H$, we are only interested in the time period from $0$ to $H$. Therefore, at time $H+1$, we know that $E[\tilde{P}_H^{N_{H+1}}(q)]=V^*_{H+1}(q)=0$ for any $q\in Q$ from Line 1 of Algorithm \ref{alg:ams}. Then for $i=H$, by Lemma \ref{lemma:1}, we know that for any $q\in Q$, it holds that
\begin{align}
&\lim_{N_{H}\rightarrow\infty}E[\tilde{P}_{H}^{N_{H}}(q)]={\max_{a\in A}(R(q,a)}\nonumber\\&{+\sum_{q'}P(q,a,q')E[\tilde{P}_{H+1}^{N_{H+1}}(q')])=\max_{a\in A}R(q,a)}=V^*_H(q).\nonumber
\end{align}

Suppose it holds for an arbitrary $0<i<H-1$, that
\begin{align}
\lim_{N_i\rightarrow\infty}\dots\lim_{N_{H}\rightarrow\infty}E[\tilde{P}_{i}^{N_{i}}(q)]=V^*_{i}(q).\nonumber
\end{align}
Consequently for $i-1$, it holds that
\begin{align}
&\lim_{N_i\rightarrow\infty}\dots\lim_{N_{H}\rightarrow\infty}E[\tilde{P}_{i-1}^{N_{i-1}}(q)]=\nonumber\\&\max_{a\in A}(R(q,a)\nonumber+\sum_{q'}P(q,a,q')E[\tilde{P}_{i}^{N_{i}}(q')])\\&=\max_{a\in A}(R(q,a)+\sum_{q'}P(q,a,q')E[V^*_{i}(q')])=V^*_{i-1}(q).\nonumber
\end{align} 
Then by induction, we know that 
\begin{align}
{\lim_{N_0\rightarrow\infty}\lim_{N_1\rightarrow\infty}\dots\lim_{N_H\rightarrow\infty}E[\tilde{P}_0^{N_0}(q)]}  = {V^*_{0}(q)}\nonumber\\ {=P^q_{max}(B_{safe}U^{\leq H}G).} \nonumber   
\end{align}

\end{proof}
Once the optimal value function (probability) has been estimated by Algorithm \ref{alg:ams}, it is then possible to extract the policy at each state $q$ and horizon $i$ by $$\mu_i(q)^*=\argmax_a \sum_{q'\in Q}P(q,a,q)\tilde{V}_i^*(q'),$$ where $\tilde{V}_i^*(q)=\tilde{P}_i^{N_i}(q)$.

This sampling approach is from the given POMDP model, which can be obtained from history data, for example, the database of medical diagnosis. Therefore, at Line (\ref{algorithm:sample1}) and  Line (\ref{algorithm:sample2}) of Algorithm \ref{alg:ams}, we sample from known distributions as defined by the POMDP model $\mathcal{P}$, instead of actually trying medication actions to  patients and observe their reactions.

\subsection{Monte Carlo Tree Search}

\begin{figure}[t]
		\begin{algorithm}[H]
  \begin{algorithmic}[1]
			\STATE{\textbf{INPUT:} {  A state $q=(b_s,e)$ in $\mathcal{B}$, time horizon $i$, the number of samples $N$, update rate $\gamma$.}}
			\STATE{\textbf{OUTPUT:} A sequence of actions $(a_h)$.}
			\STATE {$q \gets q_0$}
			\STATE {$N^q \gets 0$}
			\STATE {$\mathcal{T} \gets \{q_0$\}}
			\FOR{$h \gets 1,\dots, H$}
    			\FOR{$i \gets 1,\dots,N$}
    				\WHILE{$q \in \mathcal{T}$} 
                        \STATE {$q \gets Select(q, \mathcal{T})$}
    			    \ENDWHILE
    			    \STATE{$(q, \mathcal{T}) \gets Expand(q, \mathcal{T})$}
    			    \STATE{$R \gets Simulate(q)$}
    				\WHILE{$q \in \mathcal{T}$ } 
    				    \STATE {$q \gets$ parent of $q$}
                            \STATE {$Backpropogate(q, R)$}
    			    \ENDWHILE
        	    \ENDFOR
                \STATE {$a_h \gets \argmax_{a}(\tilde{Q}(q_0,a))$}
                \STATE {$q' \sim P(\cdot \mid q,a)$}
                \STATE {$c \gets C(q,a,q')$}
            \ENDFOR
            \STATE{\textbf{FUNCTION:}{$Select(q,\mathcal{T})$}}
            \begin{ALC@g}
                \FOR{$q' \in $ children of $q$}
                    \IF{$N^{q'} = 0$}
                        \RETURN $q'$
                    \ENDIF
                \ENDFOR
                \STATE{$q'=\argmax_{q' \in \text{children of $q$}}(\frac{\tilde{Q}(q')}{N^{q'}}+\sqrt{\frac{2\ln{n}}{N^{q'}}})$}
                \RETURN {$q'$}\;
             \end{ALC@g}
            \STATE{\textbf{FUNCTION:}{$Expand(q)$}}
            \begin{ALC@g}
                \STATE{ $\bar q \gets$ parent of $q$}
                \STATE{$N^{q} \gets 0$}
                \STATE{$Q(q) \gets Q(\bar q)$}
            \end{ALC@g}
            \STATE{\textbf{FUNCTION:}$Simulate(q, c, i)$}
            \begin{ALC@g}
                \WHILE{$q$ is not terminal}
                    \STATE{choose random action $a$}
                    \STATE{$q' \sim P(\cdot \mid q,a)$}
                    \STATE{$i \gets i + 1$}
                    \STATE{Let $q' = (b', e')$}
                    \IF{$b' \in G$}
                        \RETURN 1
                    \ENDIF
                    \IF{$e' > D$ or $b' \not\in B_{safe}$}
                        \RETURN {0}
                    \ENDIF
                \ENDWHILE
                \RETURN {0}
            \end{ALC@g}
            \STATE{\textbf{FUNCTION:} $Backpropogate(q, R)$}
            \begin{ALC@g}
                \STATE{$N^q \gets N^q + 1$}
                \STATE{$\tilde{Q}(q) \gets \gamma (R - \tilde{Q}(q))$}
            \end{ALC@g}
			\end{algorithmic}
		\caption{Monte-carlo Tree Search (MCTS)}\label{alg:mcts}			
		\end{algorithm}
	\end{figure}

In this subsection, we propose an alternative heuristic to compute policies for active classification through the use of Monte Carlo tree search (MCTS) \cite{coulom2006mcts,browne2012survey}. 
Although CB-AMS does not require computing $P_{max}(B_{safe}U^{\leq H}G)$ exactly, it does suffer from exponential time complexity dependent on the number of samples $N_i$. We relax this sample complexity by adapting the CB-AMS algorithm to the realm of MCTS. 

MCTS is a class of algorithms that solving MDPs and POMDPs by combining a tree search with random sampling. We perform a MCTS over the finite belief-state MDP $\mathcal{B}$ where the value of nodes estimates the likelihood of successful classification from that belief state. MCTS algorithms are based four main components:
\begin{enumerate}
    \item \textit{Select}: Successively select children starting from the root $R$ of the tree until reaching a leaf node $L$.
    \item \textit{Expand}: Create node $C$ as a new child of $L$.
    \item \textit{Simulate}: Complete a random rollout from $C$.
    \item \textit{Backpropagate}: Use the result of the rollout to update nodes from $C$ back to $R$.
\end{enumerate}
The full MCTS algorithm for the active classification problem is given in Algorithm \ref{alg:mcts}.

The \textit{Selection} stage requires choosing children nodes (i.e., actions) with the most promise.
We use the same UCB exploration and exploitation tradeoff that we used for the CB-AMS algorithm. Namely, children nodes (i.e., actions) in the MCTS are chosen to maximize $\frac{\tilde{Q}(q)}{N^q}+\sqrt{\frac{2\ln{n}}{N^q}}$.

If the node found by the \textit{Selection} stage does not terminate the game, the \textit{Expand} stage expands the tree by randomly choosing a child node $C$.

The \textit{Simulate} stage involves rolling out a random execution of actions from node $C$ until termination. The rollout terminates as a failure if the time or cost bound is reached or the belief leaves $B_{safe}$. The rollout terminates as a success once the belief satisfies the classification condition, i.e., $b \in G$.

Finally, \textit{Backpropogate} stage updates the value of the nodes traversed during the \textit{Select} and \textit{Expand} according to the outcome of the \textit{Rollout} stage. It also increments the visit count for each of these nodes.


\section{Examples}\label{sec:Simulation}
This section provides two examples to illustrate the applications of the proposed active classification framework and compare the exact and sampling-based algorithms\footnote{The code to the example can be found in shorturl.at/fjsxM}. 
\subsection{Medical Diagnosis and Treatment}
Following Example \ref{exp:1}, there are two possible diseases modeled by two MDPs $\mathcal{M}_1$ and $\mathcal{M}_2$. The transition probabilities are as shown in the following matrices (\ref{equation:transtion probability}), where $T_i(a)(j,k)=T_i(s_j,a,s_k),i\in\{1,2\}$ and $j,k\in\{1,2,3\}$. The costs are as defined in (\ref{equation:cost at each state}) where $C(i,j)=C(s_i,a_j)$.

\begin{myequation}\label{equation:transtion probability}
\begin{split}
T_1(a_1)=\begin{bmatrix}
   0.8 &0.2 &0 \\
   0.7&0.2&0.1 \\
   0  &0 &1
\end{bmatrix},
T_1(a_2)=\begin{bmatrix}
   0.6 &0.4 &0 \\
   0.2 &0.4    &0.4 \\
   0 &0  &1
\end{bmatrix},\\
T_1(a_3)=\begin{bmatrix}
   0.5 &0.5 &0 \\
   0.1 &0.6    &0.3 \\
   0 &0  &1
\end{bmatrix}.
T_2(a_1)=\begin{bmatrix}
   0.6 &0.4 &0 \\
   0.1&0.5&0.4 \\
   0&0&1
\end{bmatrix},\\
T_2(a_2)=\begin{bmatrix}
   0.9 &0.1 &0 \\
   0.8 &0.1    &0.1 \\
   0 &0  &1
\end{bmatrix},
T_2(a_3)=\begin{bmatrix}
   0.3 &0.7 &0 \\
   0.1 &0.3   &0.6 \\
   0 &0  &1
\end{bmatrix},
\end{split}
\end{myequation}

\begin{myequation}\label{equation:cost at each state}
C=\begin{bmatrix}
   2 & 5 & 0 \\
   6 & 4 & 0 \\
   7 & 7 & 0
\end{bmatrix}.
\end{myequation}

The diagnosis decision is made for disease $1$ or $2$ if
$
    b_s(1)\geq \lambda_1\text{ or } b_s(2)\geq \lambda_2,
$
with the initial belief $\hat{b}_{s_1}=(0.5,0.5)$, cost constraint $D=10$. It is also desirable to diagnose the disease without reaching the late stage of the disease, where the corresponding $B_{safe}$ can be defined as
\begin{equation}{\label{equation:bsafe}}
    B_{safe}=\{b_s|s\neq s_3\}. 
\end{equation}
One step unfolding according to the Algorithm \ref{alg:unfold} is shown in Figure \ref{fig:unfolding 1} from the initial belief $\hat{b}_{s_1}=(0.5,0.5)$ and cost $c=0$. With three possible actions to take, there are six subsequent states in total, two for each action. The belief is updated with \eqref{equation:HMMDP belief update} and the cost is incremented according to \eqref{equation:cost at each state}. If $\lambda_1=0.8,\lambda_2=0.7$, it can be seen that if $a_2$ is executed, there is  $0.25$ probability that the disease is diagnosed to be type $1$ (since $b_{s_2}(1)=0.8$) at the shaded state $q_4$, with a cost of $5$ since $C(s_1,a_2)=5$. Therefore, $q_4$ will not be included in the states to be expanded in the next iteration. The unfolding will then start from $\{q_2,q_3,q_5,q_6\}$ for the next round.

	\begin{figure}
		\centering	
\begin{tikzpicture}[shorten >=1pt,node distance=2.5cm,on grid,auto, bend angle=20, thick,scale=0.7, every node/.style={transform shape}] 
				\node[state] (s0)   {$q_0$}; 
				\node[state] (s1) [left =of s0] {$q_1$}; 
				\node[state] (s2) [above left  = of s0]  {$q_2$}; 
				\node[state] (s3) [above right= of s0] {$q_3$}; 
				\node[state,shade] (s4) [right= of s0] {$q_4$}; 
				\node[state] (s5) [below left= of s0] {$q_5$}; 
				\node[state] (s6) [below right= of s0] {$q_6$}; 
                \node[text width=3cm] at (-4.5,0) {$b_{s_1}=(\frac{4}{7},\frac{3}{7}),c=2$};		
                \node[text width=3cm] at (-4,1.8) {$b_{s_2}=(\frac{1}{3},\frac{2}{3}),c=2$};
                \node[text width=4cm] at (4.3,1.8) {$b_{s_1}=(0.4,0.6),c=5$};
                \node[text width=4cm] at (5,0) {$b_{s_2}=(0.8,0.2),c=5$};
                \node[text width=4cm] at (-4.2,-1.8) {$b_{s_1}=(0.625,0.375),c=0$};;
                \node[text width=4cm] at (4.5,-1.8) {$b_{s_2}=(\frac{5}{12},\frac{7}{12}),c=0$};;
                \draw [->] (0,1) -- (s0);
				\path[->]
				
				(s0) edge node {$a_1,0.7$} (s1) 
				(s0) edge node {$a_1,0.3$} (s2) 
				(s0) edge node {$a_2,0.75$} (s3) 
				(s0) edge node {$a_2,0.25$} (s4) 
				(s0) edge node {$a_3,0.4$} (s5) 
				(s0) edge node {$a_3,0.6$} (s6) 
				; 

				\end{tikzpicture} 
		\caption{One step unfolding}\label{fig:unfolding 1}
	\end{figure}

We use C++ to implement both Algorithm \ref{alg:unfold} and Algorithm \ref{alg:ams} and Python to implement Algorithm \ref{alg:mcts}. For Algorithm \ref{alg:unfold}, the resulting MDP model is input into the PRISM \cite{kwiatkowska2011prism} model checker\footnote{PRISM is a probabilistic mode checking tool that can model and analyze the quantitative probabilistic behaviors for Markov chains, MDPs, and POMDPs. The property specification includes the temporal logic, quantitative specifications, and costs/rewards.} to compute the maximum probability (\ref{equation:reachability1}). For Algorithm \ref{alg:ams}, we set $N_i=300,0\leq i\leq H$. We also store the calculated values of $\tilde{P}_i^{N_i}(q)$ to avoid recomputing them. 
For Algorithm \ref{alg:mcts} we set $N=100$.
The results are as shown in Figure \ref{fig:prob1}. We can observe  that the maximum probability to safely reach $G$ without going out of $B_{safe}$ increases with a longer time horizon. However, due to the extra safety constraint, the maximum probabilities also decrease, compared to the results without safety constraint. The CB-AMS algorithm performs well to estimate the optimal probability.

\begin{figure}
    \centering
    \includegraphics[width=3.2in]{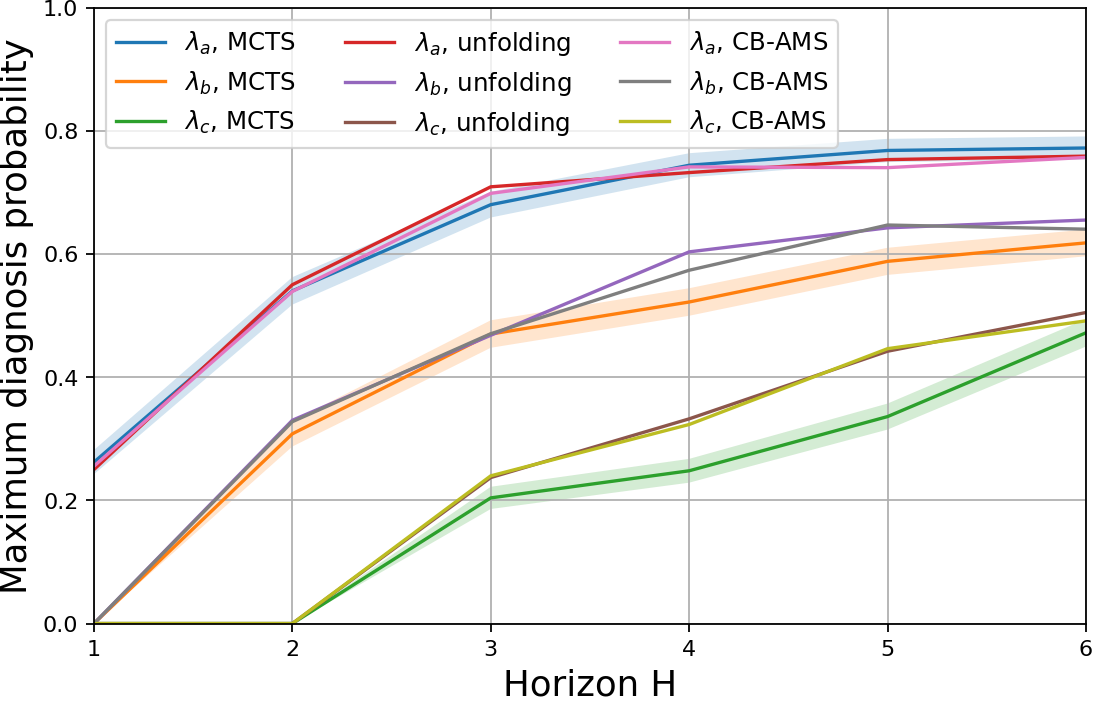}
    \caption{Maximum probability  for (\ref{equation:safe reachability}) with $B_{safe}$ defined in (\ref{equation:bsafe}). Error bars denote a 90\% confidence interval. $\lambda_a:\{\lambda_1=0.8,\lambda_2=0.7\}$, $\lambda_b:\{\lambda_1=0.9,\lambda_2=0.8\}$, $\lambda_c:\{\lambda_1=0.95,\lambda_2=0.9\}$.}
    \label{fig:prob1}
\end{figure}

We summarize the run times for both algorithms with regards to specification (\ref{equation:safe reachability}) in Table \ref{table:2}. For Algorithm \ref{alg:unfold}, the run time consists of the time to get the belief MDP model $\mathcal{B}$ and compute the optimal probability. All the experiments were run on a laptop with 2.6GHz i7 Intel\textsuperscript{\textregistered} processor and 16GB memory.  For a small time horizon $H$, exact solution outperforms sampling in time consumption, as the number of the states in $\mathcal{B}$ is small. With a growing horizon $H$, a sampling based method is favorable since its run time increases much slower.

		
		
		

\begin{table}[h]
	\centering

	\caption{Run times for (\ref{equation:bsafe}) in seconds.}
 \begin{adjustbox}{width=\columnwidth}
	\begin{tabular}{|l|l|l|l|l|l|l|l|}
		\hline
		\multicolumn{2}{|l|}{Horizon $H$} & $1$ & $2$ & $3$ & $4$ & $5$ & $6$ \\ 
		\hline
	    \multirow{2}{*}{$\lambda_a$} 
	    & Unfold & 0.40 & 0.71 & 2.39 & 7.62 &30.34 & 229.89 \\\cline{2-8}
		& CB-AMS & 0.03& 0.15 & 0.66 & 2.12 & 5.76 & 12.76 \\ \cline{2-8}
		 & MCTS & 1.74 & 3.24 & 4.26 & 5.22 & 5.64 & 6.15  \\
		\hline
		\multirow{2}{*}{$\lambda_b$} 
		& Unfold & 0.01 &	0.92 &	3.04 &	12.31 &	94.13 &	1054.71 \\\cline{2-8}
		 & CB-AMS & 0.03 & 0.15& 0.80 & 4.0 & 13.52 & 42.31 \\ \cline{2-8}
		 & MCTS & 1.60 & 3.71 & 5.28 & 6.39 & 7.25 & 7.33   \\
		\hline
		\multirow{2}{*}{$\lambda_c$} 
		& Unfold & 0.01 & 0.12 & 4.35 & 17.76 & 180.36 & 2211.70 \\\cline{2-8}
		 & CB-AMS & 0.04 & 0.21 & 0.92 & 3.95 & 14.3681 & 46.26  \\ \cline{2-8}
		 & MCTS & 1.61 & 3.97 & 5.82 & 7.64 & 8.80 & 9.83   \\
		\hline
	\end{tabular}
 \end{adjustbox}
	\label{table:2}
\end{table}

\begin{figure}[t]
	\vspace{-3mm}
		\subfloat[\label{fig:privacy mdp}]{
	\begin{tikzpicture}[shorten >=1pt,node distance=2cm,on grid,auto, thick,scale=0.95, every node/.style={transform shape}]
	\node[state] (q_1)   {$s_1$};
	\node[state] (q_2) [below left = 2cm of q_1] {$s_2$};
	\node[state] (q_3) [below right =2cm of q_1] {$s_3$};
	
	\path[->]
	(q_1) edge [pos=0.5, loop above] node  {} (q_1)
	(q_1) edge [pos=0.5, bend right, above=0.5,sloped] node {} (q_2)
	(q_1) edge [pos=0.5, bend left, above=0.5,sloped] node {} (q_3)
	
	(q_2) edge [pos=0.5, loop below, left=0.1] node  {} (q_2)
	(q_2) edge [pos=0.5, bend right, above=0.5,sloped] node {} (q_1)
	(q_2) edge [pos=0.5, above=0.5] node {} (q_3)
	
	(q_3) edge [pos=0.5, loop below, right=0.1] node {} (q_3)
	(q_3) edge [pos=0.5, bend left, above=0.5,sloped] node {} (q_1)
	(q_3) edge [pos=0.5, bend left, above=0.5,sloped] node {} (q_2)
	;
	\end{tikzpicture}
		}
		\subfloat[\label{fig:ads}]{
\centering
	\begin{tikzpicture}[scale=0.63]
	\draw[very thick,->] (0,0) -- (5.2,0) node[above] {$b(\mathcal{C}_1=1)$};
	\draw[very thick,->] (0,0) -- (0,5.2) node[right] {$b(\mathcal{C}_2=1)$};
	
	\draw (0,0.05) -- (0,-0.05) node[below] {\footnotesize 0};
	\draw (1,0.05) -- (1,-0.05) node[below] {\footnotesize 0.2};
	\draw (2,0.05) -- (2,-0.05) node[below] {\footnotesize 0.4};
	\draw (3,0.05) -- (3,-0.05) node[below] {\footnotesize 0.6};
	\draw (4,0.05) -- (4,-0.05) node[below] {\footnotesize 0.8};
	\draw (5,0.05) -- (5,-0.05) node[below] {\footnotesize1};
	
	\draw (-0.05,1) -- (0.05,1) node[left] {\footnotesize 0.2};
	\draw (-0.05,2) -- (0.05,2) node[left] {\footnotesize 0.4};
	\draw (-0.05,3) -- (0.05,3) node[left] {\footnotesize 0.6};
	\draw (-0.05,4) -- (0.05,4) node[left] {\footnotesize 0.8};
	\draw (-0.05,5) -- (0.05,5) node[left] {\footnotesize 1};
	\fill[blue!50!cyan,opacity=0.3] (0,0) -- (1,0) -- (1,5) -- (0,5) -- cycle;
	\fill[blue!50!cyan,opacity=0.3] (4,0) -- (5,0) -- (5,5) -- (4,5) -- cycle;
	\fill[red!50!cyan,opacity=0.3] (0,0) -- (5,0) -- (5,1.25) -- (0,1.25) -- cycle;
	\fill[red!50!cyan,opacity=0.3] (0,3.75) -- (5,3.75) -- (5,5) -- (0,5) -- cycle;
	\draw[dashed] (0,1.25) -- (1,1.25) -- (1,3.75) -- (0,3.75) -- (0,1);
	\draw[dashed]  (4,1.25) -- (5,1.25) -- (5,3.75) -- (4,3.75) -- (4,1);
	\draw[black,fill=black] (2.5,2.5)  circle (.5ex);
	\draw[very thick,->] (2.5,2.5) -- (3,2.5) -- (2.3,3.2) -- (2,2) -- (1.5,2.8) -- (1.2,2.3) -- (0.8,3);
	\draw[very thick,->,dashed] (2.5,2.5) -- (2.3,2.8) -- (2,1.5) -- (3.7,2) -- (3.8,2.7) --(3.9,3.2)-- (4.2,4);
	\end{tikzpicture}
		}
	\caption{(a) shows the feasible state transitions in the MDPs. (b) illustrates the belief space in the interactive advertising example. The area surrounded by dashed lines is $G$. The solid and dashed trajectories represent  successful and failed attempts to classification, respectively.}
\end{figure}

\subsection{Intruder Classification}
In automated surveillance applications \cite{bharadwaj2017synthesis}, mobile sensors or robots are deployed to monitor the intruder of potential threats.  For example, the work in \cite{bharadwaj2017synthesis} synthesizes controllers for mobile sensors with quantitative surveillance requirements. However, it is often the case that detected targets are not threats to start with, such as when small animals or disturbances set off alarms.  it is often necessary to determine whether a detected target is a potential threat before deploying security resources for further intervention. In these cases, we want to monitor the situation and determine whether a detected target is a threat before we deploy the mobile sensors to perform active surveillance. We assume we can always passively observe the target's location, which is generally possible through radar or some other static sensors in the environment.

\begin{figure}
\centering
\subfloat{
\includegraphics[scale=0.25]{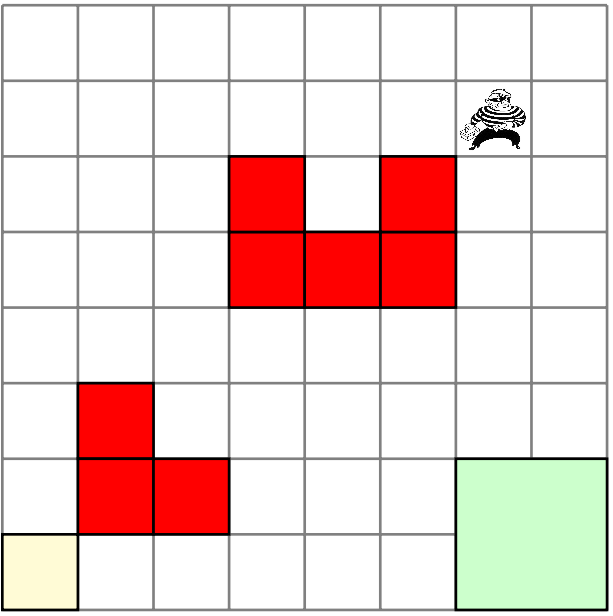}}
\subfloat{
 \includegraphics[scale=0.25]{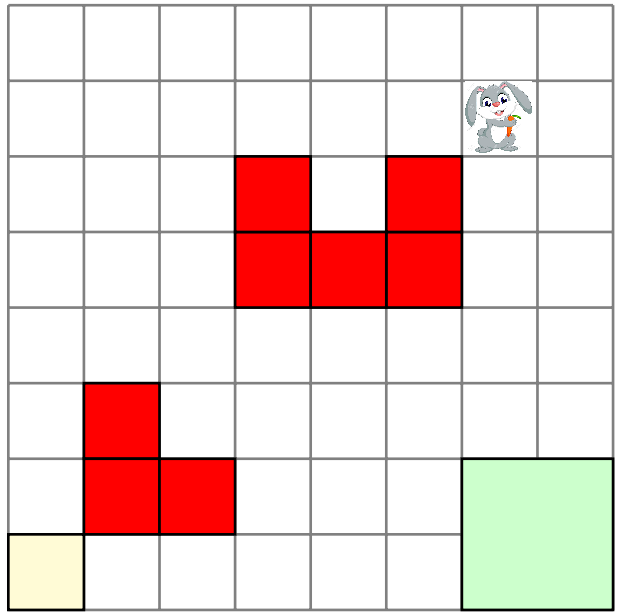}}
		\caption{Gridworld with two potential classes of targets - hostile (left) and safe (right). Green cells are sensitive areas. Red cells are obstacles such as buildings. The yellow cell is a hiding place for the hostile intruder. } \label{fig:grid}
\end{figure}
An example of a $8\times 8$ gridworld is shown in Figure~\ref{fig:grid}. The target is not allowed to reach the green zone. The target is assumed to be either  a hostile human intruder (Class 1) or an animal (Class 2) that has no real threat. Therefore, like the first example, there is only one attribute in $\mathcal{C}$. The behaviors of the hostile and safe intruders are characterized by two MDPs $\mathcal{M}_1$ and $\mathcal{M}_2$, respectively. The state  in the MDPs refers to the target's location in the gridworld, which is  observable through radar or some other static sensors in the environment.  

\begin{figure*}[ht]
	\centering\includegraphics[scale=0.42]{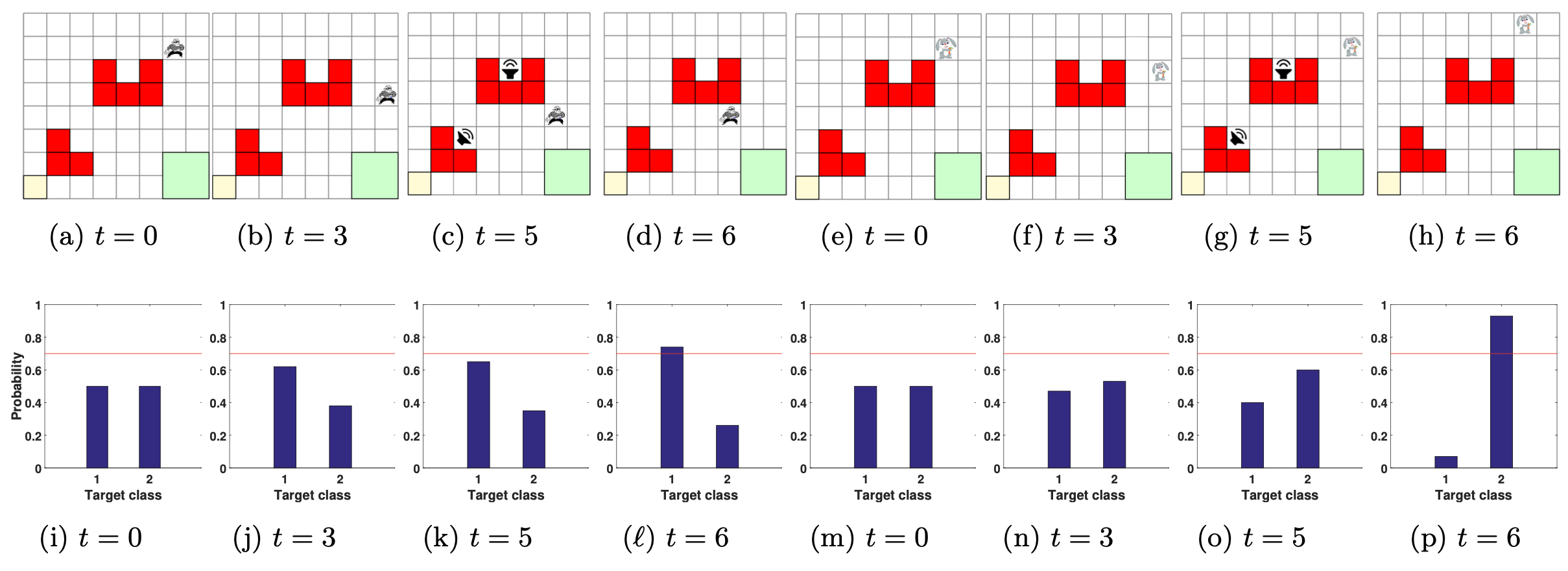}

	\caption{Simulation of the classification problem in a gridworld environment with both target classes. Figures 7a-7d and7i-7l 
 correspond to a simulation with a hostile target. Figures7e-7h and7m-7p correspond to a simulation with a safe target.}
	\label{fig:gridresults}
\end{figure*}
At each time step, the target randomly moves to one of its neighboring cells. For example, if the MDP model of class 1 assumes that the target will take action South from the current state, we assume there is a small probability of the target taking one of the actions North, East, or West instead, parametrized by  a randomness parameter $0 \leq \gamma \leq 1$. If $\gamma = 0$, there is no uncertainty  and if $\gamma = 1$, we know nothing about the target's behavior which means it can take all actions with equal probability. The actions available for the automated surveillance are $a_1$ for passive observation and $a_2$ for alarm through  loudspeakers.

If $a_1$ is chosen, a human intruder will attempt to move to the sensitive area (the green region) as denoted by $S_{green}$. When $a_2$ is executed, the animal will be startled and  move in all directions with equal probability, while the human intruder will tend to move towards the yellow region  in Figure \ref{fig:grid} ostensibly to hide. The human moves randomly  but generally heads to  $S_{green}$ or yellow region for action $a_1$ or $a_2$. The randomness is to capture different human preferences and  human inherent decision uncertainty. The costs for $a_1$ and $a_2$ at each state are $1$ and $3$, respectively.

The corresponding $B_{safe}$ can be defined as
$
    B_{safe}=\{b_s|s\notin S_{green}\}. 
$
The classification decision is  made  if one of the following  is satisfied:
$
    b_s(1)\geq 0.9\text{ or } b_s(2)\geq 0.9,
$
with the initial belief $\hat{b}_{\hat{s}}=(0.5,0.5)$, step bound  $H=6$ and cost bound $D=8$, where $\hat{s}$ is the initial state that intruders at as seen in Figure \ref{fig:grid}.




Figure~\ref{fig:gridresults} illustrates two instances of the classification process, where the exact solution is used to compute the optimal classification strategy. Suppose the target is a hostile intruder, Figures~7a- 7d illustrate a run of the human movement, where action $a_1$ is executed at $t=0,3,6$ and  action $a_2$ is executed at $t=5$. Figure~7i to -7l 
depicts the corresponding beliefs. In the scenario shown in Figure 7c, the target is near $S_{safe}$ and the corresponding belief is  as seen in Figure~7k which favors class $1$. The optimal action at the time instance $t=5$ is to sound the alarm. Then at $t=6$, it is observed that the target moves towards the yellow cell and the belief of class $1$ exceeds the threshold. This is where the classification terminates and a human operator will be alerted. Figures~7e to 7h and 7m to 7p shows the classification with a safe target (an animal) where similar behavior can be observed. After the alarm is used in $t=5$ in Figure 7o, the rabbit runs to the top right corner as seen in Figure 7p, a very unlikely move for a hostile intruder. As a result, the belief for class $2$ exceeds the threshold and the classification process terminates. In both simulations, the final costs are $7$.   In some cases, it is possible to classify with only passive observations. This situation usually occurs when $\gamma$ is small which means the uncertainties in candidate models are relatively low. However, as $\gamma$ is increased, more observations are needed, and purely relying on passive observations may not be possible if we need to classify the target before it reaches the green region.

We evaluate the performance of Algorithm \ref{alg:unfold} (unfolding), Algorithm \ref{alg:ams} (AMS) and Algorithm \ref{alg:mcts} (MCTS) with regard to specification (5) in the intruder environment with $\gamma=0.1$. For Algorithm \ref{alg:unfold}, the run time consists of the time to get the belief MDP model $\mathcal{B}$ and compute the optimal probability. Algorithm \ref{alg:unfold} takes 6.01 seconds total and computes $\mu^* = 0.802$.  Algorithm \ref{alg:ams} takes 0.40 seconds with an estimated maximal probability of $0.79$. Algorithm \ref{alg:mcts} takes 0.074 seconds with an estimated maximal probability of $0.79$ (averaged over 100 samples).

\subsection{Wildlife Classification}

Unmanned aerial vehicles (UAVs) equipped with cameras have been shown to be effective tools for carrying out automated wildlife detection through low-disturbance aerial surveys \cite{gonzalez2016wildlife,lopez2019drones}. Traditionally, these methods are divided into two phases: 1) a data acquisition phase where the UAV follows a predetermined flight path to collect images and 2) a classification phase where the data is analyzed to detect and count wildlife \cite{gonzalez2016wildlife}. In this example, we demonstrate how this methodology could be extended to an active classification setting where online classification can help inform the data collection.

The problem takes place on an $N \times N$ gridworld containing randomly placed obstacles. The gridworld contains a single UAV along with 2 unknown animals with randomized starting locations that have a 50\% prior probability of being a Kangaroo. At each time step, the UAV camera can move in one of the four the cardinal directions while the unknown animals move either vertically or horizontally back-and-forth across the gridworld.

After each time step, the UAV receives noisy observations of the type of each unknown animal based on the distance between the UAV and the animal. If the UAV's vision of the animal is obstructed by an obstacle, the observation is totally random. Otherwise, the probability that the UAV correctly classifies an animal at a distance of $D$ as either being a Kangaroo or not  is
 $   \max(0.95 - e^{-4.0/D}, 0.5)$
meaning that an observation has a maximum of 95\% chance of being correct and at a distance of over 5 units is totally random.

We only evaluate Algorithm \ref{alg:mcts} on the wildlife classification task, since the other approaches are intractable over the large state space. Since there are two unknown animals, each of which can be classified as an animal or not, there are four possible true underlying MDPs. Classification thresholds are 0.9 for each of the underlying possible MDPs. The whole belief space is considered safe and there is no cost constraint. We evaluate Algorithm \ref{alg:mcts} with $N=100$. Figure \ref{fig:wildlife_scale} shows the scalability and Figure \ref{fig:wildlife_success} shows the success probability of Algorithm \ref{alg:mcts} with respect to varying sizes of gridworlds and a growing horizon. Both plots are obtained by sampling 100 random seeds.

\begin{figure}[ht]
	\subfloat[Average computation time. \label{fig:wildlife_scale}]{
		\includegraphics[scale=0.5]{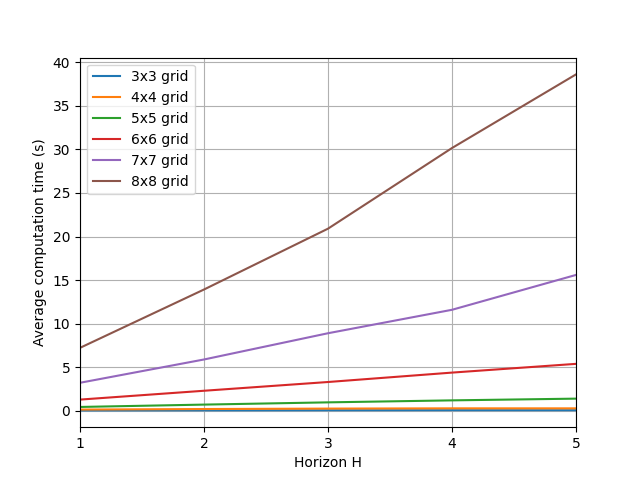}
	}
	\hfill
	\subfloat[Estimated likelihood of success. \label{fig:wildlife_success}]{
		\includegraphics[scale=0.5]{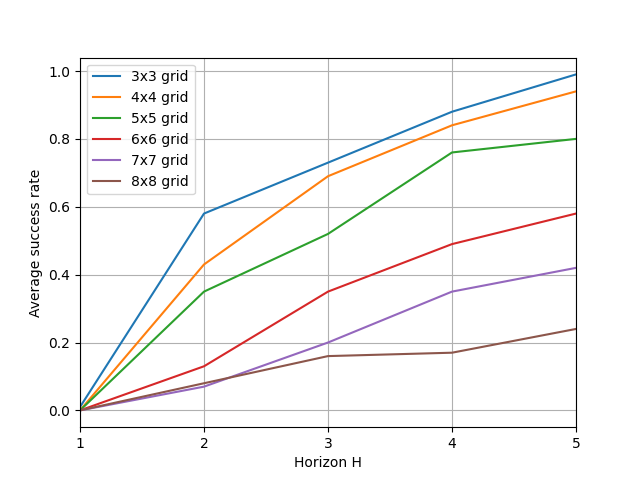}
	}
	\caption{The average computation time \ref{fig:wildlife_scale} and estimated likelihood of success \ref{fig:wildlife_success} of Algorithm \ref{alg:mcts} on increasing horizon and varying gridworld size for the wildlife classification task.}
\end{figure}

\section{Conclusion}\label{sec:Conclusion}
In this paper, we studied a cost-bounded active classification of certain attributes of dynamical systems belonging to a  finite set of MDPs. We utilized the HMMDP modeling framework and the objective was to actively select actions based on the current belief, accumulated cost, and time step, {such that the probability to reach a classification decision within a cost bound can be maximized while avoiding unsafe belief states}. To solve the problem, we proposed three approaches. The first one was an exact solver to obtain the unfolded belief MDP model considering the cost-bound, and then compute the optimal strategy. To mitigate the computation burden, the rest two approaches adaptively sample the actions to estimate the maximum probability offline and online, respectively. Three examples are given to show the application of our proposed approaches. 

For future work, it is of interest to study how the performance, in terms of the maximum probability to reach a classification decision without visiting unsafe belief regions, deteriorates in the approximate solution with the number of samples. Furthermore, the POMDP in this paper has a special structure where the underlying MDP, once selected, will not change. We would also like to further study how to leverage this fact to reduce the computation complexity.



\bibliographystyle{elsarticle-harv}
\bibliography{ref}
\end{document}